\documentclass[reqno,11pt]{amsart}

\newtheorem{definition}{Definition}

\usepackage{graphicx}
\usepackage{xcolor}
\usepackage{tikz}
\usepackage{tikz-cd}
\usepackage{caption}
\usepackage{booktabs}
\usepackage{adjustbox}
\usepackage{longtable}
\usepackage{array}
\usepackage{amssymb}

\usepackage{amsmath}

\newcommand{\be}{\begin{equation}}
\newcommand{\ee}{\end{equation}}

\DeclareMathSymbol{\Lambda}{\mathord}{operators}{"03}

\usepackage{enumerate}
\usepackage{enumitem}

\usepackage{appendix}

\newtheorem{thm}{Theorem}[section]
\newtheorem{prop}[thm]{Proposition}
\newtheorem{theorem}[thm]{Theorem}

\newtheorem{lemma}[thm]{Lemma}

\newtheorem{remark}[thm]{\it Remark}
\newtheorem{example}[thm]{\it Example}

\usepackage{tikz-3dplot}
\usepackage{xifthen}
\usepackage{caption}

\newcommand{\cp}{\mathbb C\mathbb P^1}

\tdplotsetmaincoords{60}{125}
\tdplotsetrotatedcoords{0}{20}{0} 

\usepackage[
labelfont=sf,
hypcap=false,
format=hang
]{caption}

\usepackage{mathtools}

\begin{document}

\title{On quadrirational pentagon maps}

\author[C. Evripidou]{Charalampos Evripidou}
\address{Charalampos Evripidou, Department of Mathematics and Statistics, University of Cyprus, P.O. Box 20537, 1678, Nicosia, Cyprus}
 \email{evripidou.charalambos@ucy.ac.cy}

\author[P. Kassotakis]{Pavlos Kassotakis}
\address{Pavlos Kassotakis, Department of Mathematical Methods in Physics, Faculty of Physics,
University of Warsaw, Pasteura 5, 02-093, Warsaw, Poland}
 \email{Pavlos.Kassotakis@fuw.edu.pl, pavlos1978@gmail.com}

\author[A. Tongas]{Anastasios Tongas}
\address{Anastasios Tongas,  Department of Mathematics, University of Patras, 26 504 Patras, Greece}
 \email{tasos@math.upatras.gr}

\date{\today}

\begin{abstract}
We classify rational solutions of a specific type  of the set theoretical version of the pentagon equation. That is,
we find all quadrirational maps $R:(x,y)\mapsto (u(x,y),v(x,y)),$ where $u, v$ are two rational functions on two arguments,
that serve as  solutions of the pentagon equation. Furthermore, provided a  pentagon map that admits a partial inverse, we
obtain genuine entwining pentagon set theoretical solutions. Finally, we show how to obtain Yang-Baxter maps from entwining pentagon maps.
\end{abstract}

\maketitle

\section{Introduction}
The {\em pentagon equation} serves as one of a handful of very important equations in mathematical physics. It firstly appeared  as an identity satisfied by the  Racah coefficients \cite{Biedenharn:1953,Elliott:1953}. 
Also it was introduced in \cite{Drinfeld_p}  as a consistency relation in quasi-Hopf
algebras,   in  \cite{MoorSeib:89} inside the context of conformal field theory and in \cite{Maillet:1990,Kaufmann_1993}
in connection with three-dimensional integrable systems. In detail the pentagon equation (see Figure \ref{fig01}) reads
\begin{gather} \label{pe_def}
 R_{12}R_{13}R_{23}=R_{23}R_{12}.
\end{gather}
In the equation above, $R$ can stand for an operator, so we have the operator version of (\ref{pe_def}); or $R$ can stand
for a map so we have the set theoretic version of the latter. In the set theoretic version of (\ref{pe_def}) the subscripts
denote the sets where the maps $R_{ij}$ act nontrivially, while in the operator version of (\ref{pe_def}), they denote the vector
spaces where non-trivial action takes place. Solutions of the set theoretic version of of (\ref{pe_def}) are called {\em set theoretical solutions of the pentagon equation} or simply {\em pentagon maps}.

The pentagon equation and its {\em dual } equation  that reads
\begin{gather} \label{pe_def_d}
 R_{23}R_{13}R_{12}=R_{12}R_{23},
\end{gather}
and is referred to as {\em reverse-pentagon equation} (see Figure \ref{fig02}), serve as the
first non-trivial examples of the so-called polygon equations  \cite{Dimakis:2015,korepanov:2022,Hoissen2023}.
Note that both equations above are similar to the well known Yang-Baxter equation $R_{23}R_{13}R_{12}=R_{12}R_{13}R_{23},$
with the middle term missing on either side.

Shortly after the pentagon equation was introduced, set theoretical solutions  in connection with Poisson maps appeared in \cite{Zakrzewski:1992},
and in \cite{Skandalis:1993} in connection with measured spaces.  A  systematic study of set theoretical solutions of the
pentagon equation can be found in  \cite{Kashaev:1998}. Furthermore, set theoretical solutions of the pentagon equation and of its dual are interrelated with various areas of Mathematics and Physics for instance
with geometric topology \cite{Korepanov:2000}, with incidence geometry \cite{Doliwa:2014p}, with functional analysis
\cite{Faddeev:1994,Kashaev:1999}, with discrete integrable systems \cite{Doliwa:2020}, just to name a few. For a survey on set theoretical solutions of the pentagon equation we
refer to \cite{Mazzotta:2023}. For further
connections and interrelations of the pentagon equation  we refer to  \cite{Dimakis:2015,Hoissen2023} where also its
combinatorial  structures as well as associated linear problems (cf. \cite{Kassotakis:2023_p}) were studied.

  To classify pentagon maps on  sets without any structure seems beyond reach. Nevertheless,
if some structure is imposed on the sets, classification results do exist. Indeed, in \cite{Catino:2020} pentagon maps on groups
were classified in terms of their normal subgroups.

\begin{figure}[htb]\adjustbox{scale=0.6,center}{
\begin{minipage}[htb]{0.2\textwidth}
\begin{tikzcd}[row sep=0.5cm, column sep = 0.5cm,every arrow/.append style={dash}]
 1 \arrow[rightarrow,from=1-1,to=2-2]& {} & {} & {} &{} \\
 2 \arrow[rightarrow,from=2-1,to=2-2] & R_{12}\arrow[rightarrow,from=2-2,to=2-4] \arrow[rightarrow,from=2-2,to=3-3]  & {}& R_{23}
 \arrow[rightarrow,from=2-4,to=2-5] \arrow[rightarrow,from=2-4,to=1-5]& {} \\
 {} & {} & R_{13} \arrow[rightarrow,from=3-3,to=4-5] \arrow[rightarrow,from=3-3,to=2-4] &{}  & {} \\
 3 \arrow[rightarrow,from=4-1,to=3-3] &{} &{} &{} &{}
\end{tikzcd}
\end{minipage} \hspace{4cm} = \hspace{2cm}
\begin{minipage}[htb]{0.2\textwidth}
\begin{tikzcd}[row sep=0.5cm, column sep = 0.5cm,every arrow/.append style={dash}]
 1 \arrow[rightarrow,from=1-1,to=3-4]& {} & {} & {} &{} \\
 {} & {} &  {}  &{}  & {} \\
 2 \arrow[rightarrow,from=3-1,to=3-2] & R_{23}\arrow[rightarrow, crossing over,from=3-2,to=1-5] \arrow[rightarrow,from=3-2,to=3-4] &
 {} & R_{12} \arrow[rightarrow,from=3-4,to=3-5] \arrow[rightarrow,from=3-4,to=4-5] &{} \\
 3 \arrow[rightarrow,from=4-1,to=3-2] &{} &{} &{} &{}
\end{tikzcd}
\end{minipage}}
\caption{The pentagon equation}\label{fig01}
\end{figure}
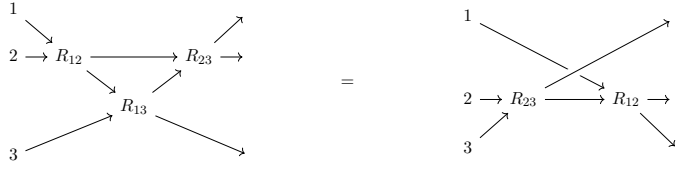
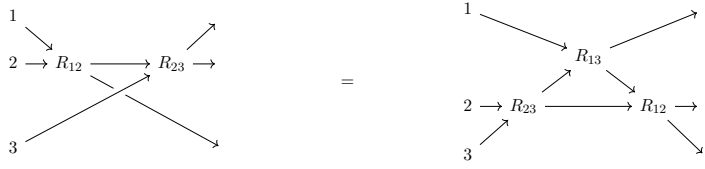
\begin{figure}[htb]\adjustbox{scale=0.6,center}{
\begin{minipage}[htb]{0.2\textwidth}
\begin{tikzcd}[row sep=0.5cm, column sep = 0.5cm,every arrow/.append style={dash}]
 1 \arrow[rightarrow,from=1-1,to=2-2]& {} & {} & {} &{} \\
 2 \arrow[rightarrow,from=2-1,to=2-2] & R_{12}\arrow[rightarrow,from=2-2,to=2-4] \arrow[rightarrow,crossing over, from=2-2,to=4-5]  & {}& R_{23} \arrow[rightarrow,from=2-4,to=2-5] \arrow[rightarrow,from=2-4,to=1-5]& {} \\
  {} & {} &   &{}  & {} \\
 3 \arrow[rightarrow,crossing over, from=4-1,to=2-4] &{} &{} &{} &{}
\end{tikzcd}
\end{minipage} \hspace{4cm} = \hspace{2cm}
\begin{minipage}[htb]{0.2\textwidth}
\begin{tikzcd}[row sep=0.5cm, column sep = 0.5cm,every arrow/.append style={dash}]
 1 \arrow[rightarrow,from=1-1,to=2-3]& {} & {} & {} &{} \\
  {} & {} &  R_{13} \arrow[rightarrow,from=2-3,to=3-4] \arrow[rightarrow,from=2-3,to=1-5] &{}  & {} \\
   2 \arrow[rightarrow,from=3-1,to=3-2] & R_{23}\arrow[rightarrow, from=3-2,to=2-3] \arrow[rightarrow,from=3-2,to=3-4] & {} & R_{12} \arrow[rightarrow,from=3-4,to=3-5] \arrow[rightarrow,from=3-4,to=4-5] &{} \\
   3 \arrow[rightarrow,from=4-1,to=3-2] &{} &{} &{} &{}
\end{tikzcd}
\end{minipage}}
\caption{The reverse pentagon equation}\label{fig02}
\end{figure}

This article is organized as follows. We begin with a brief introduction. In Section \ref{sec2}, we present our main result of this article, that is we propose a classification scheme for {\em quadrirational} pentagon maps, modulo an equivalence relation. In that respect we find four non-trivial representatives of equivalence classes of such maps characterized by their singularity sets. For
all these maps we also give the corresponding refactorization problem in terms of Lax matrices. In Section \ref{sec3} we show how the companion maps (partial inverses) of a quadrirational map entwine to provide {\em genuine entwining} (non-constant)  pentagon set theoretical solutions. In addition, we construct Yang-Baxter maps associated with quadrirational pentagon maps. 

\section{Quadrirational  pentagon maps} \label{sec2}

In this Section we propose a classification scheme for pentagon maps of a specific type.
We  give a full list of representatives, that satisfy the pentagon equation, modulo an equivalence relation that
is defined on birational functions on $\cp\times\cp$.

The basic definitions and properties are given in a more general setting.

\begin{definition}
A map $R$ is called   {\em pentagon map}
if it is a solution to the set-theoretic version of the  pentagon equation (\ref{pe_def}).
\end{definition}

Let $\mathbb{X}$ be a set. There is a natural equivalence relation on pentagon maps
$R:\mathbb{X}\times \mathbb{X}\rightarrow\mathbb{X}\times \mathbb{X}$.
\begin{definition}
Two  maps $R: \mathbb{X} \times \mathbb{X} \rightarrow \mathbb{X} \times \mathbb{X}$ and
$S: \mathbb{X} \times \mathbb{X} \rightarrow \mathbb{X} \times \mathbb{X}$ are called
$M\ddot{o}b$ equivalent if there exists a bijection $\phi: \mathbb{X} \rightarrow \mathbb{X}$ such that
$ R\circ  (\phi\times \phi)= (\phi\times \phi)\circ S.$
\end{definition}
In detail, if we write the maps $R, S$ in terms of their components as $R=(R_1, R_2)$ and $S=(S_1, S_2)$ then they are $M\ddot{o}b$ equivalent
if there exists a bijection $\phi: \mathbb{X} \rightarrow \mathbb{X}$ such that
$$
R_1(\phi(x),\phi(y))=\phi(S_1(x,y)) \quad \text{and} \quad R_2(\phi(x),\phi(y))=\phi(S_2(x,y))\,.
$$
The importance of the equivalence relation defined above is given in the following proposition.

\begin{prop}
Let $R: \mathbb{X} \times \mathbb{X} \rightarrow \mathbb{X} \times \mathbb{X}$ be a pentagon map and $S$
a $M\ddot{o}b$ equivalent map to $R$. Then $S$ is also a pentagon map.
\end{prop}
\begin{proof}
Recall that $R_{ij}:\mathbb{X} \times \mathbb{X} \times \mathbb{X} \rightarrow \mathbb{X} \times \mathbb{X} \times \mathbb{X}$
acts on the $i,j$ components as $R$ and on the rest as identity. Therefore, for any bijection $\phi: \mathbb X\rightarrow\mathbb X$,
the map $(\phi^{-1}\times\phi^{-1}\times\phi^{-1})\circ R_{ij}\circ(\phi\times\phi\times\phi)$ acts on the $i,j$ components as the conjugation
of the components of $R$ and on the rest as the identity.
It follows that if $S=(\phi^{-1}\times\phi^{-1})\circ R\circ(\phi\times\phi)$ then
\begin{gather*}
S_{12}\circ S_{13}\circ S_{23}= (\phi^{-1}\times \phi^{-1}\times \phi^{-1})\circ R_{12}\circ  R_{13}\circ R_{23}\circ
(\phi \times \phi \times \phi)\\
=(\phi^{-1}\times \phi^{-1}\times \phi^{-1})\circ R_{23} \circ R_{12}\circ (\phi\times \phi\times \phi)=
S_{23} \circ S_{12},
\end{gather*}
where we used the hypothesis that $R$ is a pentagon map.
\end{proof}

\begin{remark} \label{remark:1}
Let $\tau:\mathbb{X} \times \mathbb{X}\mapsto \mathbb{X} \times \mathbb{X}$ be the transposition map i.e. $\tau: (x,y)\mapsto (y,x).$  Then $R$ is a  pentagon map iff $T:=R\circ \tau$ satisfies the {\em braid-pentagon} equation
\begin{gather}\label{bp}
T_{12}\circ T_{23}\circ T_{12}=T_{23}\circ \tau_{12}\circ T_{23}.
\end{gather}
Clearly the  equation above is equivalent to the pentagon equation (\ref{pe_def}) and  it was in this form that the pentagon equation was  introduced in \cite{MoorSeib:89}.

Furthermore, it is easy to verify the following:
\begin{itemize}
\item a map $R: \mathbb{X} \times \mathbb{X} \rightarrow \mathbb{X} \times \mathbb{X}$ is a pentagon map iff $\tau\circ R\circ \tau$ is a reverse-pentagon map;
\item an invertible map     $R: \mathbb{X} \times \mathbb{X} \rightarrow \mathbb{X} \times \mathbb{X}$ is a pentagon map iff $R^{-1}$ is a reverse-pentagon map;
\end{itemize}
\end{remark}

In the following we restrict our considerations to rational maps of $\cp\times \cp$ into itself, therefore  $\mathbb X=\cp,$ which we identify with
$\mathbb C\cup \left\{ \infty \right\}$ with its usual operations. We recall the notion of {\em quadrirational maps} which
first appeared in \cite{Etingof_1999} under the name {\em non-degenerate rational maps} and shortly afterwards, in \cite{ABS:YB},
was renamed quadrirational maps.

\begin{definition}[\cite{Etingof_1999,ABS:YB}]\label{quadrirational}
A  map $R: \mathbb{X} \times \mathbb{X}\ni(x,y)\mapsto (u,v) \in \mathbb{X} \times \mathbb{X}$ is called quadrirational,
if both the map $R$ and the so called {\em companion map} (or {\em partial inverse})
$cR: \mathbb{X} \times \mathbb{X}\ni(x,v)\mapsto (u,y) \in \mathbb{X} \times \mathbb{X}$ are birational maps.
\end{definition}
Said differently, the birational map $R=(u,v)$ is quadrirational if for any $y\in \mathbb X$ (generic), the map $u(.,y): x\mapsto u(x,y)$ is birational 
and for any $x\in \mathbb X$ (generic), the map $v(x,.): y\mapsto v(x,y)$ is also birational.
In \cite{ABS:YB}, the quadrirational maps were
classified and it was shown that necessarily have the form  
\begin{align*}
u(x,y)=&\frac{a(y)x+b(y)}{c(y)x+d(y)},& v(x,y)=&\frac{A(x)y+B(x)}{C(x)y+D(x)},
\end{align*}
for some polynomials $a, b, \dots, D$ which are at most quadratic.

The notion of quadrirational maps can be extended to maps that act on the cartesian product $\mathbb{X} \times\ldots\times \mathbb{X}$
($n$ copies). Examples of such maps   can be found in \cite{2n-rat}, where they are referred to as {\em $2^n-$rational maps}.

Let
$R:(x,y)\mapsto (u(x,y),v(x,y))$
be a pentagon map,
that is it satisfies the pentagon equation
\begin{align}\label{pentagon}
R_{12}\circ R_{13}\circ R_{23}=R_{23}\circ R_{12}\,.
\end{align}
Then its components $u, v$ necessarily satisfy the following relations
\begin{align} \label{eq1}
  u(x, y)&=u\left(u(x, v(y, z)), u(y, z)\right),\\ \label{eq2}
  u(v(x, y), z)&=v\left(u(x, v(y, z)), u(y, z)\right),\\ \label{eq3}
  v(v(x, y), z)&=v(x, v(y, z)).
\end{align}

We immediately recognize that (\ref{eq3}) says that $v$ is an associative function. We have from \cite{Brawley} the following theorem.
\begin{theorem}[\cite{Brawley}] \label{theo0}
If $v:\cp\times\cp\rightarrow\cp$ is a nonconstant associative rational function then there exists a M\"obius transformation
$\phi:\cp\rightarrow\cp$ such that $\phi^{-1}\circ v\circ (\phi\times\phi)$
is equal to $x, y, x+y$ or $x + y- xy.$
\end{theorem}
According to the theorem above, the second argument of a pentagon map can be one of the four functions stated 
or any $M\ddot{o}b$ equivalent function to them. Since two pentagon maps are $M\ddot{o}b$ equivalent if both of their arguments are
conjugate under the same M\"obius transformation, it is natural to examine which M\"obius transformations fix the associative functions
provided in Theorem \ref{theo0}. This is addressed in the following lemma.

\begin{lemma}
Let $\phi:\cp\rightarrow\cp$ be a M\"obius transformation.
\begin{enumerate}
\item For the functions $v=x$ or $v=y$ we have $\phi^{-1}\circ v\circ\phi\times \phi=v$ (for any $\phi$).
\item For $v=x+y,\; \phi^{-1}\circ v\circ(\phi\times \phi)=v$ if and only if $\phi(x)=ax$ for some $a\in\cp$ .
\item For $v=x+y-xy,\; \phi^{-1}\circ v\circ(\phi\times \phi)= v$
if and only if $\phi(x)=x$ or $\phi(x)=\frac{x}{x-1}$.
\end{enumerate}
\end{lemma}
\begin{proof}
The proof of item (1) is straightforward.
The proof of item (2) follows easily by observing that such a $\phi$ is linear, i.e.
$\phi(x)+\phi(y)=\phi(x+y)$ for all $x, y \in\cp$.
For the proof of item (3) observe that if $\psi=\frac{x-1}{x},$ then  $\psi^{-1}\circ v\circ(\psi\times \psi)=xy$ and the only M\"obius transformations
$\kappa:\cp\rightarrow\cp$ that fix $xy$ (i.e. $\kappa(xy)=\kappa(x)\kappa(y)$) are $\kappa(x)=x$ and $\kappa(x)=\frac{1}{x}$. Therefore
$\phi(x)=x$ or $\phi(x)=\frac{x}{x-1}$.
\end{proof}

Notice that for $\delta\neq 0$, if $\phi_\delta(x)=\frac{x-1}{\delta x}$ and $v_\delta=x+y-\delta xy$, then $\phi_\delta^{-1}\circ v_\delta\circ (\phi_\delta\times\phi_\delta)=xy$.
Thus the functions $v_\delta$ are all $M\ddot{o}b$ equivalent with each other and with $xy$. Clearly they are all associative.
Furthermore, by Definition \ref{quadrirational} a map with second argument $x$ cannot be qudrirational as its companion map is not birational in this case.
Therefore for any quadrirational pentagon map $R=(u,v)$, there are only three possibilities for $v$, namely $v=y$ and $v=x+y-\delta xy$ with $\delta=0, 1$.

The following theorem is the main result of the paper.




\begin{theorem} \label{theo1}
Any quadrirational pentagon map $R:\cp\times \cp\rightarrow\cp\times\cp$, with $R=(u,v)$ is $M\ddot{o}b$ equivalent to exactly one
of the following  maps:
\begin{align*}
&u= \frac{x}{x+y-x y }, && v=x+y-x y, && (S_I)\\
&u=x ,&& v=x+y-\delta x y, && (S_{II}^\delta)\\
&u=\frac{x}{y}, && v=y, && (S_{III}) \\
&u=x-y, && v=y, && (S_{IV}) 
\end{align*}
where $\delta=0,1$.
\end{theorem}

\begin{proof}
In order for a rational map $R:(x,y)\mapsto (u,v)$ to be a pentagon map, the functional equations (\ref{eq1})--(\ref{eq3}) have to be satisfied.
The functional equation (\ref{eq3}) is satisfied when $v$ is an associative rational function. In what follows, for any of the
representatives (except $v(x,y)=x,$ that does not give quadrirational maps) of associative rational functions given in Theorem \ref{theo0}, we find all rational functions $u$ that satisfy the equations
(\ref{eq1}) and (\ref{eq2}).

We begin with the case where $v(x,y)=x+y-\delta xy$ ($\delta\neq 0$), which as we saw is equivalent with $x+y-xy$.
Then equation (\ref{eq3}) is identically satisfied while (\ref{eq1}) and (\ref{eq2}) respectively read
\begin{align}\label{f11}
u(x,y)&=u\left(u(x,y+z-\delta yz),u(y,z)\right)\\ \label{f12}
u(x+y-\delta xy,z)&=u(x,y+z-\delta yz)+u(y,z)-\delta u(x,y+z-\delta yz)u(y,z).
\end{align}
Because of the quadrirationality of $R$, the rational function $u$ is of the form
$
u(x,y)=\frac{a(y)x+b(y)}{c(y)x+d(y)},
$
where the polynomials $a, b, c$ and $d$ are at most quadratic in $y$. Substituting this form of $u$ into the equations (\ref{f11}) and
(\ref{f12}), we obtain a system of algebraic equations on the coefficients of the polynomials $a, b, c$ and $d$.
Solving this system we obtain that $u$ is either $x$ or $\frac{1}{\delta}\frac{x}{x+y-\delta x y }$ or $-\frac{x}{\delta}\frac{1-\delta y}{y}.$
The first choice  $u=x$, leads exactly  to  $(S_{II}^\delta)$ pentagon map.
From the second choice $u=\frac{1}{\delta}\frac{x}{x+y-\delta x y }$, we obtain the map
\begin{align*}
&\hat S_I:(x,y)\mapsto (u,v)=\left(\frac{1}{\delta}\frac{x}{x+y-\delta x y },x+y-\delta x y\right),
\end{align*}
which is $M\ddot{o}b$ equivalent to the  $(S_{I})$ pentagon map. Indeed, under the scaling
\begin{gather*}
\phi:\mathbb{CP}^1 \ni u\mapsto \frac{1}{\delta} u \in \mathbb{CP}^1,
\end{gather*}
we have
\begin{gather*}
S_I=(\phi^{-1}\times \phi^{-1})\circ \hat S_I\circ(\phi\times \phi).
\end{gather*}
The third choice of $u$ gives rise to the pentagon map
\begin{align*}
&\hat s_I:(x,y)\mapsto (u,v)=\left(-\frac{x}{\delta}\frac{1-\delta y}{y},x+y-\delta x y\right),
\end{align*}
which is $M\ddot{o}b$ equivalent to $\hat S_I$ and therefore to the pentagon map $S_I$ of the theorem as well.
Indeed for
\begin{gather} \label{mob00}
\phi:\mathbb{CP}^1 \ni u\mapsto -\frac{u}{1-\delta u} \in \mathbb{CP}^1\,,
\end{gather}
we have
\begin{gather*}
\hat s_I=(\phi^{-1}\times \phi^{-1})\circ \hat S_I\circ (\phi\times \phi)\,.
\end{gather*}

Note that the M\"obius transformations used above to show the equivalence of the functions $u$, fix the functions $v$, in accordance to
Theorem \ref{theo0}. They where derived by examining the sets of singular points of  $\hat S_I$ and $\hat s_I$ which are respectively
\begin{align*}
\Sigma_{\hat S_I^a}&=\{(0,0), (\infty,1/\delta ),(1/\delta ,\infty)\},&
\Sigma_{\hat s_I^a}&=\{(0,0), (1/\delta ,\infty), (\infty,1/\delta )\}.
\end{align*}
The M\"obius transformation (\ref{mob00}) fixes $(0,0)$ and permutes the remaining singular points of the previous sets.

Let us now consider the case  $v(x,y)=y$. In this case, equation (\ref{eq3}) is (in accordance to Theorem \ref{theo0}) satisfied,
while equation (\ref{eq2}) reduces to the identity $u(y,z)=u(y,z)$.  So we only have (\ref{eq1}) which reads
\begin{align}
\label{f21}
u(x,y)&=u\left(u(x,z),u(y,z)\right)\,.
\end{align}
Because of the quadrirationality of $R$, and therefore the specific form of $u$ used above, we obtain as before the following seven
pentagon maps $R_i:(x,y)\mapsto (u_i,y),$ $i=1,\ldots, 7$
\begin{align*}
u_1&=x\,, &
u_2&=x-y+a\,, &
u_3&=\frac{xy}{-x+y+bxy}\,,& \\
u_4&=\frac{x}{y}\,, &
u_5&=\frac{x-y}{1-ay-bxy} \; (|a|+|b|\neq0)\,, &
u_6&=\frac{x-axy}{ax+by-a^2xy-abxy}\; (ab\neq 0)\,,& \\
u_7&=\mathrlap{\frac{ax-by-cxy+1}{cx-b^2y+aby-cy-bcxy+b}\; (ab\neq c)\,,}
\end{align*}
where $a,b,c\in \mathbb{CP}^1\setminus\{\infty\}$ are constants.
As we mentioned after Theorem \ref{theo0}, $v=y$ is fixed under any M\"obius transformation. Consider the following M\"obius transformations
\begin{align*}
\phi_{2,3}(x,y)&=\frac{1+(a-b)x}{x}\,,&
\phi_{4,5}(x,y)&=\frac{1-(a+r)x}{1+rx}\;,& \phi_{4,6}(x,y)&=\frac{bx}{1-ax}\,,\\
\phi_{4,7}(x,y)&=\mathrlap{\frac{1}{c-ab}\frac{(b^3+rb^2-ab^2+2bc+rc)x-(b^2+c+ra+rb)}{1+rx}\,,}
\end{align*}
where for  $\phi_{4,5},$ $r$ is defined by $r^2+ar-b=0$ and for  $\phi_{4,7},$ $r$ is defined by $r^2+(b-a)r-c=0$.
It is straightforward to verify that $\phi_{i,j}^{-1}\circ u_i\circ(\phi_{i,j}\times \phi_{i,j})=u_j$
for the $\phi_{i,j}, u_i$ and $u_j$ given above. We will show that $u_1, u_2$ and $u_4$ are not $M\ddot{o}b$ equivalent.
Since $u_1=x$ depends on $x$ only, it is fixed under any invertible function on $\cp$ and therefore it is
only equivalent to itself. We show that $\frac{x}{y}$ and $x-y$ are not equivalent by any invertible function
$\phi:\cp\rightarrow\cp$. Such a function has to satisfy $\phi(x-y)=\frac{\phi(x)}{\phi(y)}$ for all
$x,y\in\cp$ and therefore $\phi(x)=e^{\lambda x}$ which is not a bijection. This completes the proof.
\end{proof}
In the previous theorem we have not considered the case $v(x,y)=x,$ since it results to non-quadrirational maps.
Furthermore, for $v(x,y)=x$ equation (\ref{eq3}) is identically satisfied, while (\ref{eq1}) and (\ref{eq2}) respectively read
\begin{align} \label{id1}
u(x,y)&=u(u(x,y),u(x,y)),\\ \label{id2}
u(x,z)&=u(x,y).
\end{align}
From (\ref{id2}) we deduce that $u$ is a function of a single variable i.e. $u(x,y)=g(x).$ Then $(\ref{id1})$ reads
$g(x)=g(g(x)),$ so $g(x)$ has to be idempotent. So we obtain the following rational (but not birational)
map that satisfies the pentagon equation
\begin{gather*}
T:(x,y)\mapsto (g(x),x), \qquad \text{where} \qquad g(g(x))=g(x).
\end{gather*}
This is a simple  example where idempotency appears in relation of pentagon maps. For  general treatment on idempotency and pentagon
maps we refer to \cite{Mazzotta_2023}, while for involutivity and pentagon maps, see \cite{Colazzo_2023}.

Also we have not considered the trivial solution $R(x,y)=(x,y)$ which appeared in the proof of the previous theorem.

The following remarks are in order.
\begin{itemize}
\item The inverse maps $S_{I-IV}^{-1}$ of Theorem \ref{theo0} satisfy the reverse pentagon equation (\ref{pe_def_d}),
while the mappings $S_{I-IV}\circ \tau$ satisfy the braid-pentagon equation (\ref{bp}).
\item The sets of singular points of the mappings $S_{I-IV}$ respectively are
\begin{align*}
 \Sigma_{S_I}&=\{(0,0), (\infty,1),(1,\infty)\}, &\Sigma_{S_{II}^\delta}&=\{(\infty,1/\delta),(1/\delta,\infty)\},\\
 \Sigma_{S_{III}}&=\{(0,0),(\infty,\infty)\}, &\Sigma_{S_{IV}}&=\{(\infty,\infty)^2\}.
\end{align*}
The singularity pattern of the maps confirms the non equivalence of the four representatives of the classification in Theorem \ref{theo1}.
Furthermore, mappings $S_{II}^\delta, S_{III}$ and $S_{IV}$ can be obtained from $S_I$ by a degeneration procedure (see Figure \ref{fig3}).
Indeed, $S_{III}$ is obtained from $S_I$ by setting $(u,x,v,y)\mapsto (\epsilon u, \epsilon x,v,y)$ and then sending $\epsilon \rightarrow 0$.
Mapping $S_{II}^1$ is derived from $S_I$ by setting $(u,x,v,y)\mapsto (u,x,1-\epsilon(1-v),1-\epsilon(1-y))$ and then sending
$\epsilon \rightarrow 0,$ while by setting $(u,x,v,y)\mapsto (\epsilon u,\epsilon x, \epsilon v, \epsilon y),\; \epsilon \rightarrow 0,$
we obtain $S_{II}^0$ from $S_{II}^1$. Finally, by setting
$(u,x,v,y)\mapsto (1+\epsilon u,1+\epsilon x,\frac{1}{1-\epsilon v},\frac{1}{1-\epsilon y}),\; \epsilon\rightarrow0$
we obtain $S_{IV}$ from $S_{III}$.
\begin{figure}[htb]\adjustbox{scale=0.9,center}{
\begin{tikzcd}[row sep=0.6cm, column sep = 0.8cm,every arrow/.append style={dash}]
{}& S_{II}^1  \arrow[rightarrow,from=1-2,to=1-3]& S_{II}^0 \\
S_I \arrow[rightarrow,from=2-1,to=1-2]  \arrow[rightarrow,from=2-1,to=3-2]&{} & {}\\
{}&S_{III} \arrow[rightarrow,from=3-2,to=3-3] & S_{IV}
\end{tikzcd}}
\caption{Degeneration diagram}\label{fig3}
\end{figure}
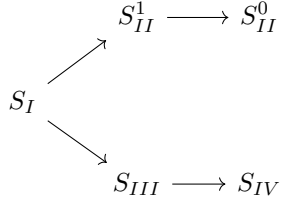

\item The results of the theorem  can be  extended to the non-abelian setting by considering $\mathbb{X}=\mathbb{A}^\times$
a division ring e.g. bi-quaternions. Or we can consider a more general setting where $\mathbb{A}^\times$  stands for the
subgroup of invertible matrices of the algebra $\mathbb{A}$ of  matrices of order $N$. In detail, the non-abelian
analogues of the pentagon maps of Theorem \ref{theo0} read
\begin{align*}
&u= x(x+y-y x)^{-1},&& v=x+y-yx, && (\mathfrak{S}_I)\\
&u=x ,&& v=x+y-\delta y x, && (\mathfrak{S}_{II}^\delta)\\
&u=xy^{-1} ,&& v=y, && (\mathfrak{S}_{III}) \\
&u=x-y ,&& v=y, && (\mathfrak{S}_{IV})
\end{align*}
where $\delta=0,1$.
\item Mapping $S_I$ was firstly introduced in \cite{Kashaev:1999} inside the context of quantum dilogarithm.
Furthermore, $S_I$  also results from the evolution of matrix KP solitons \cite{Dimakis:2018}. The non-abelian form of
$S_I$ that is  $\mathfrak{S}_I,$ arises as a reduction of the so-called {\em normalization map} \cite{Doliwa:2014p,Doliwa:2020},  that also  allows the transition to solutions of the operator form of the pentagon equation. In \cite{Doliwa:2014p}, $\mathfrak{S}_I,$ was also related to the non-abelian Hirota-Miwa equation.  Mapping $\mathfrak{S}_{II}^\delta$ (in an equivalent form) first appeared in \cite{Kashaev:1998}.
\item When a pentagon map is a linear map we can associate a linear operator to it. In that respect, solutions that satisfy the operator version of the pentagon equation can be obtained by the linearization of the pentagon maps. Actually, through this procedure  the so-called {\em permutation-type solutions} of the pentagon equation are obtained. These type of solutions were first introduced  in \cite{Hietarinta:1997} inside the context of the quantum Yang-Baxter equation.
\item The   pentagon maps of Theorem \ref{theo0}, satisfy the so-called {\em ten-term relation} \cite{Kashaev:1998}, therefore    they correspond to four equivalence classes of non-Abelian tetrahedron maps i.e. maps from $T:\mathbb{A}^\times\times\mathbb{A}^\times\times\mathbb{A}^\times$ to itself that satisfy $T_{123}\circ T_{145}\circ T_{246}\circ T_{356}=T_{356}\circ T_{246} \circ T_{145} \circ T_{123}.$ For recent developments on tetrahedron maps we refer to \cite{Kass1,Dimakis_Korepanov:2021,Talalaev_2021,Rizos:2022,Rizos:2023}.
\end{itemize}

\subsection{Lax matrices for the quadrirational pentagon maps}
In \cite{Kassotakis:2023_p} there were introduced pairs of matrices of a specific form that their associated refactorization
problems turned equivalent to pentagon maps. These pairs of matrices serve as the   {\em Lax pairs} of the associated
discrete dynamical systems.

\begin{prop} \label{refa}
The pentagon maps $\mathfrak{S}_I, \mathfrak{S}_{II}^\delta,$ $\mathfrak{S}_{III}$ and $\mathfrak{S}_{IV},$ that serve as the non-Abelian analogues of the
respective maps of Theorem \ref{theo1}, are equivalent to the refactorization problems
\begin{gather}\label{ZC}
  A(u)B(v)=B(y)A(x),
\end{gather}
where the matrices $A$ and $B$ respectively read
\begin{align*}
&A(x):=\begin{pmatrix}
        1-x & x \\
        0 & 1
      \end{pmatrix},&& B(x):=\begin{pmatrix}
        1 & 0 \\
        1-x & x
      \end{pmatrix}, && (\mathfrak{L}_I)\\
    &  A(x):=\begin{pmatrix}
        1-\delta x & 0 \\
        0 & 1
      \end{pmatrix},&& B(x):=\begin{pmatrix}
        1 & 0 \\
        1-\delta x & 1
      \end{pmatrix}, && (\mathfrak{L}_{II}^\delta)\\
&A(x):=\begin{pmatrix}
        1 & x \\
        0 & 1
      \end{pmatrix},&& B(x):=\begin{pmatrix}
        1 & 0 \\
        0 & x
      \end{pmatrix}, && (\mathfrak{L}_{III})\\
&A(x):=\begin{pmatrix}
        1 & e^x \\
        0 & 1
      \end{pmatrix},&& B(x):=\begin{pmatrix}
        1 & 0 \\
        0 & e^x
      \end{pmatrix}, && (\mathfrak{L}_{IV})
\end{align*}
where $\delta \neq 0$.
\begin{proof}
The proof follows by direct computation.
\end{proof}
\end{prop}
 An alternative interpretation of the matrix refactorization
problem (\ref{ZC})  is the following parameter dependent
associativity condition
\begin{align*}
p \circ_x ( q \circ_y  r ) = (p \circ_u  q) \circ_v r,
\end{align*}
for $p,q,r$ vectors in some vector space $\mathcal{V}$.
This
interpretation was introduced in \cite{Hoissen2023} (c.f. \cite{Dimakis:2015})
and it is motivated by considering the zero curvature condition
as a {\em local tetragon equation}.
 Then the above associativity
condition for the binary operations defined by
\begin{align*}
 p \circ_x q&= x\,p+(1-x)\,q, & p \circ_x q&= p+(1-\delta
x)\,q,&
p \circ_x q&= x\,p+q, & p \circ_x q&= e^x\,p+q
\end{align*}
delivers the respective maps $\mathfrak{S}_{I-IV}.$ Note that when $p,q$ are considered as elements of a module over the ring 
$\mathbb{Z}[x,x^{-1}]$ of Laurent polynomials in the variable $x$, the first binary operation from above coincides with the left-distributive binary operation of the Alexander quandle.
A geometric
interpretation of the associativity condition is provided by a $(6_2,
4_3)$ configuration on the plane (the {\em Veblen or  Menelaus configuration}), which consists of 6 points and 4
lines, where each point is incident with exactly two lines, each line
with exactly 3 points, and the binary operation represents the
collinearity of three points $p$, $q$, $p\circ_x q$. Then the pentagon equation reads as a consistency condition on the {\em Desargues configuration} $(10_3)$ that contains five Veblen configurations, see
\cite{Doliwa:2014p}.

\begin{remark}
We note that $M\ddot{o}b$ equivalent pentagon maps admit $M\ddot{o}b$ equivalent refactorization problems. Indeed, let $T:(x,y)\mapsto (u,v)$ be a $M\ddot{o}b$ equivalent  pentagon map  to the pentagon map $S:(x,y)\mapsto (u,v),$ that is there exists a bijection $\phi$ such that $(\phi\times\phi)\circ S=T\circ (\phi\times\phi).$ Let also the pentagon map $S$ be equivalent to the refactorization problem $A(u)B(v)=B(y)A(x)$. Then the pentagon map $T$ is equivalent to the refactorization problem $A(\phi(u))B(\phi(v))=B(\phi(y))A(\phi(x)).$
\end{remark}
\begin{remark}
The refactorization problems (\ref{ZC}) allow us to obtain the invariant relations of the associated maps.
For example, from $\mathfrak{L}_I$ we obtain that
\begin{gather*}
 Tr\left(A(u)B(v)\right)=Tr\left(B(y)A(x)\right),
\end{gather*}
where with $Tr$ we denote the trace of a matrix, which
reads as the invariant relation
\begin{gather} \label{ab_l}
 (1-u)v=y(1-x),
\end{gather}
of mapping $\mathfrak{S}_I.$
\end{remark}
The following Proposition states that mapping $S_I$ is a Liouville integrable map \cite{ves2}.
\begin{prop}
Mapping $S_I:(x,y)\mapsto (u,v)$  preserves the
Poisson structure
$
 \Omega:=xy(1-x)(1-y)\frac{\partial}{\partial x}\wedge \frac{\partial}{\partial y},
$
and admits the function $h(x,y):=y(1-x)$ as an invariant, hence it is a Liouville integrable map.
\end{prop}
%


\section{Entwining pentagon maps}\label{sec3}

The {\em entiwining} (or {\em non-constant}) form of the set-theoretical version of the  pentagon equation reads
\begin{align}\label{pent_e}
 R^{(1)}_{12}\circ R^{(3)}_{13}\circ R^{(5)}_{23}=R^{(4)}_{23}\circ R^{(2)}_{12},
 \end{align}
where the superscripts inside the parentheses denote mappings $R^{(q)}$ that might differ. When  the maps $R^{(q)}$ do not differ
we recover the  {\em constant} pentagon equation considered in Section \ref{sec2}. It might happen that at least one of the maps
$R^{(q)}$ that participate in (\ref{pent_e}) could be a pentagon map, but this is not necessarily the case as we shall see.

Inside the context of matrix refactorization problems the first examples of entwining pentagon maps were considered in \cite{Kassotakis:2023_p}, where also the notion of
{\em genuine} entwining pentagon maps was introduced.
\begin{definition}
If  at least one pair of the maps   $R^{(q)}$ that satisfy  (\ref{pent_e}) is not $M\ddot{o}b$ equivalent,
then the maps $R^{(q)}$ will be called {\em genuine} entwining maps.
\end{definition}
We present now the main theorem of this section.
\begin{theorem} \label{theo3}
Let $R$ be a quadrirational pentagon map. Then the following entwining pentagon equations are  satisfied
\begin{gather}\label{en1}
R^{(1)}_{12}\circ R^{(3)}_{13}\circ R^{(3)}_{23}=R^{(3)}_{23}\circ R^{(1)}_{12},\\ \label{en2}
R^{(4)}_{12}\circ R^{(4)}_{13}\circ R^{(1)}_{23}=R^{(1)}_{23}\circ R^{(4)}_{12},\\ \label{en3}
R^{(4)}_{12}\circ R^{(2)}_{13}\circ R^{(3)}_{23}=R^{(3)}_{23}\circ R^{(4)}_{12},
\end{gather}
where $R^{(1)}:=R,$ $R^{(2)}:=R^{-1}$ the inverse of $R,$ $R^{(3)}:=cR$ the companion map of $R$ and $R^{(4)}:=(cR)^{-1}.$
The corresponding to the equations (\ref{en1})-(\ref{en3}) entwining pentagon maps are genuine.
\end{theorem}
\begin{proof}
The assumption that the mapping $R^{(1)}:(x,y)\mapsto \left(u(x,y),v(x,y)\right)$ is quadrirational guarantees the existence of the maps
\begin{gather}\label{quad1}
\begin{aligned}
  R^{(2)}: \left(u(x,y),v(x,y)\right)\mapsto (x,y), \\
   R^{(3)}: \left(u(x,y),y\right)\mapsto (x,v(x,y)), \\
   R^{(4)}: \left(x,v(x,y)\right)\mapsto (u(x,y),y),
\end{aligned}
\end{gather}

We prove  that $R^{(1)}$ and $R^{(3)}$ satisfy (\ref{en1}) which equivalently reads
\begin{gather*}
\left(R^{(1)}_{12}\right)^{-1}\circ R^{(3)}_{23}\circ R^{(1)}_{12}\circ \left(R^{(3)}_{23}\right)^{-1}\circ \left(R^{(3)}_{13}\right)^{-1}=id,
\end{gather*}
or
\begin{gather*}
R^{(2)}_{12}\circ R^{(3)}_{23}\circ R^{(1)}_{12}\circ R^{(4)}_{23}\circ R^{(4)}_{13}=id.
\end{gather*}
We have to show that
\begin{gather*}
R^{(2)}_{12}\left( R^{(3)}_{23}\left( R^{(1)}_{12}\left( R^{(4)}_{23}\left( R^{(4)}_{13} P_0 \right) \right) \right) \right)=P_0,
\end{gather*}
where $P_0=(x,y,z)$ a generic point. Since $R$ is quadrirational, $v$ is a bijection w.r.t the second argument.
Therefore we  may take $P_0=\left(x,y,v\left(x,v(y,z)\right)\right)$ and then
\begin{gather*}
  R^{(4)}_{13} P_0=\left(u\left(x,v(y,z)\right),y,v(y,z)\right)=:P_1,\\
  R^{(4)}_{23} P_1=\left(u\left(x,v(y,z)\right),u(y,z),z\right)=:P_2, \\
  R^{(1)}_{12} P_2=\left(u\left(u\left(x,v(y,z)\right),u(y,z)\right),v\left(u\left(x,v(y,z)\right),u(y,z)
  \right),z\right)=\left(u(x,y),u\left(v(x,y),z\right),z\right)=:P_3,
\end{gather*}
where we have used (\ref{quad1}), as well as (\ref{eq1}),(\ref{eq2}) which hold because $R^{(1)}$ is a pentagon map. Therefore
\begin{gather*}
  R^{(3)}_{23} P_3=\left(u(x,y),v(x,y),v\left(v(x,y),z\right)\right)=:P_4, \\
  R^{(2)}_{12} P_4=\left(x,y,v\left(v(x,y),z\right)\right)=\left(x,y,v\left(x,v(y,z)\right)\right)=P_0,
\end{gather*}
since the function $v$ is associative due to (\ref{eq3}) and that completes the proof.

The proofs that $R^{(1)},$  $R^{(4)}$ satisfy (\ref{en2}) and that $R^{(2)}, R^{(3)}, R^{(4)}$  satisfy (\ref{en3}) follow in a similar way.
 Indeed, (\ref{en2}) equivalently reads
\begin{gather*}
  R^{(3)}_{12}\circ R^{(2)}_{23}\circ R^{(4)}_{12}\circ R^{(4)}_{13}\circ R^{(1)}_{23}=id,
\end{gather*}
and then it is not difficult to show that
\begin{gather*}
  R^{(3)}_{12}\left( R^{(2)}_{23}\left( R^{(4)}_{12}\left( R^{(4)}_{13}\left( R^{(1)}_{23} P_0''\right)\right)\right)\right)=P_0'',
\end{gather*}
where $P_0''=\left(x,v(x,y),z\right),$ an arbitrary point.

Similarly, (\ref{en3}) equivalently reads
\begin{gather*}
  R^{(3)}_{23}\circ R^{(4)}_{12}\circ R^{(4)}_{23}\circ R^{(1)}_{13}\circ R^{(3)}_{12}=id,
\end{gather*}
and  that holds since
\begin{gather*}
  R^{(3)}_{23}\left( R^{(4)}_{12}\left( R^{(4)}_{23}\left( R^{(1)}_{13}\left( R^{(3)}_{12} P_0'''\right)\right)\right)\right)=P_0''',
\end{gather*}
where $P_0'''=\left(u(x,y),y,v(y,z)\right),$ an arbitrary point.
\end{proof}
Note that in the entwining pentagon equation (\ref{en3}),  none of the maps that participate  is a pentagon map. Furthermore,
relaxing the quadrirationality assumption of Theorem \ref{theo3}, a stronger  version of the latter  can be formulated.
\begin{theorem} \label{theo4}
Let $R:(x,y)\mapsto \left(u(x,y),v(x,y)\right)$ be a pentagon map.
\begin{enumerate}[label=\Alph*.]
\item If the partial inverse $S: \left(u(x,y),y\right)\mapsto \left(x,v(x,y)\right)$ of the map $R$ exists,  then the  entwining pentagon equation
\begin{gather*}
  R_{12}\circ S_{13}\circ S_{23}=S_{23}\circ R_{12},
\end{gather*}
is satisfied.
\item If the partial inverse $T: \left(x,v(x,y)\right)\mapsto \left(u(x,y),y\right)$ of the map $R$ exists,   then the  entwining pentagon equation
\begin{gather*}
  T_{12}\circ T_{13}\circ R_{23}=R_{23}\circ T_{12},
\end{gather*}
is satisfied.
\end{enumerate}
\end{theorem}
\begin{proof}
The proof follows in a similar manner as the proof of Theorem \ref{theo3}.
By direct computation we can show that
\begin{gather*}
 R_{12}\left( S_{13}\left( S_{23}P_0\right)\right) =S_{23}\left( R_{12} P_0\right),
\end{gather*}
where $P_0=\left(u\left(x,v(y,z)\right),u(y,z),z\right)$ and that proves the first part of the theorem.

For the second part of the theorem, it is not difficult to show that
\begin{gather*}
 T_{12}\left( T_{13}\left( R_{23}P_0\right)\right) =R_{23}\left( T_{12} P_0\right),
\end{gather*}
where $P_0=\left(x,v(x,y),z\right),$ and that completes the proof of the second part of the theorem.

Note that the choices of $P_0$ above do not affect the generality because of our assumption that $u$ is a bijection in the first argument
(in the first case) and $v$ is a bijection in the second argument (in the second case).
\end{proof}
In the following Proposition we associate Yang-Baxter maps with  entwining pentagon maps. 
\begin{prop}
Let $R$ be a quadrirational pentagon map. Then the map
\begin{align}\label{PtoYB}
\mathfrak{R}_{(1,2)(3,4)}:=R^{(3)}_{14}\circ R^{(1)}_{13}\circ R^{(2)}_{24}\circ R^{(4)}_{23},
\end{align} 
with $R^{(i)}$ given in Theorem \ref{theo3}, is a Yang-Baxter map i.e. it satisfies
\begin{align} \label{ybr}
\mathfrak{R}_{\bf{12}}\circ \mathfrak{R}_{\bf{13}}\circ \mathfrak{R}_{\bf{23}}=\mathfrak{R}_{\bf{23}}\circ \mathfrak{R}_{\bf{13}}\circ \mathfrak{R}_{\bf{12}},
\end{align}
where ${\bf 1}:=(1,2),{\bf 2}:=(3,4)$ and ${\bf 3}:=(5,6).$ 
\end{prop}
\begin{proof}
  The proof follows by re-interpreting some of the results of \cite{Kashaev:19977}  in the set theoretical setting. Indeed, since $R^{(1)}:=R,$ is a quadrirational pentagon map, equations (\ref{pe_def}) and (\ref{en1})-(\ref{en3}) are satisfied. In addition, the following reverse pentagon equations hold 
  \begin{gather} \label{yb1}
R^{(2)}_{23}\circ R^{(2)}_{13}\circ R^{(2)}_{12}=R^{(2)}_{12}\circ R^{(2)}_{23},\\ \label{yb2}
R^{(4)}_{23}\circ R^{(4)}_{13}\circ R^{(2)}_{12}=R^{(2)}_{12}\circ R^{(4)}_{23},\\ \label{yb3}
R^{(2)}_{23}\circ R^{(3)}_{13}\circ R^{(3)}_{12}=R^{(3)}_{12}\circ R^{(2)}_{23},\\  \label{yb4}
R^{(4)}_{23}\circ R^{(1)}_{13}\circ R^{(3)}_{12}=R^{(3)}_{12}\circ R^{(4)}_{23}.
\end{gather}
Then from (\ref{ybr}) we have
\begin{multline*}
R^{(3)}_{14}\circ R^{(1)}_{13}\circ R^{(2)}_{24}\circ R^{(3)}_{16}\circ R^{(1)}_{15}\circ \underbrace{R^{(4)}_{23}\circ R^{(2)}_{26}\circ R^{(3)}_{36}}_{= R^{(3)}_{36}\circ R^{(4)}_{23} \;\; \mbox{due to (\ref{en3})} }\circ R^{(4)}_{25}\circ R^{(1)}_{35}\circ R^{(2)}_{46}\circ R^{(4)}_{45}\\
- R^{(3)}_{36}\circ R^{(1)}_{35}\circ R^{(2)}_{46}\circ R^{(3)}_{16}\circ \underbrace{R^{(4)}_{45}\circ R^{(1)}_{15}\circ R^{(3)}_{14}}_{= R^{(3)}_{14}\circ R^{(4)}_{45} \;\; \mbox{due to (\ref{yb4})}}\circ R^{(2)}_{26}\circ R^{(4)}_{25}\circ R^{(1)}_{13}\circ R^{(2)}_{24}\circ R^{(4)}_{23}
\end{multline*}
\begin{multline*}
=R^{(3)}_{14}\circ R^{(1)}_{13}\circ R^{(2)}_{24}\circ R^{(3)}_{16}\circ R^{(1)}_{15}\circ R^{(3)}_{36}\circ \underbrace{R^{(4)}_{23} \circ R^{(4)}_{25}\circ R^{(1)}_{35}}_{=R^{(1)}_{35}\circ R^{(4)}_{23}\;\; \mbox{due to (\ref{en2})}}\circ R^{(2)}_{46}\circ R^{(4)}_{45}\\
- R^{(3)}_{36}\circ R^{(1)}_{35}\circ \underbrace{R^{(2)}_{46}\circ R^{(3)}_{16}\circ R^{(3)}_{14}}_{= R^{(3)}_{14}\circ R^{(2)}_{46} \;\; \mbox{due to (\ref{yb3})}}\circ R^{(4)}_{45}\circ R^{(2)}_{26}\circ R^{(4)}_{25}\circ R^{(2)}_{24}\circ R^{(1)}_{13}\circ R^{(4)}_{23}
\end{multline*} 
\begin{multline*}
=R^{(3)}_{14}\circ R^{(2)}_{24}\circ \underbrace{R^{(1)}_{13}\circ R^{(3)}_{16}\circ R^{(3)}_{36}}_{=R^{(3)}_{36}\circ R^{(1)}_{13} \;\;\mbox{due to (\ref{en1})}}\circ R^{(1)}_{15}\circ R^{(1)}_{35}\circ R^{(4)}_{23}\circ R^{(2)}_{46}\circ R^{(4)}_{45}\\
- R^{(3)}_{36}\circ R^{(1)}_{35}\circ R^{(3)}_{14}\circ R^{(2)}_{46}\circ R^{(2)}_{26} \circ \underbrace{ R^{(4)}_{45}\circ R^{(4)}_{25}\circ R^{(2)}_{24}}_{=R^{(2)}_{24}\circ R^{(4)}_{45}\;\;\mbox{due to (\ref{yb2})}}\circ R^{(1)}_{13}\circ R^{(4)}_{23}
\end{multline*} 
\begin{multline*}
=R^{(3)}_{14}\circ R^{(2)}_{24}\circ R^{(3)}_{36}\circ \underbrace{R^{(1)}_{13} \circ R^{(1)}_{15}\circ R^{(1)}_{35}}_{=R^{(1)}_{35}\circ R^{(1)}_{13}\;\;\mbox{due to (\ref{pe_def})}}\circ R^{(4)}_{23}\circ R^{(2)}_{46}\circ R^{(4)}_{45}\\
- R^{(3)}_{36}\circ R^{(1)}_{35}\circ R^{(3)}_{14}\circ \underbrace{R^{(2)}_{46}\circ R^{(2)}_{26} \circ R^{(2)}_{24}}_{=R^{(2)}_{24}\circ R^{(2)}_{46}\;\;\mbox{due to (\ref{yb1})}}\circ R^{(4)}_{45}\circ R^{(1)}_{13}\circ R^{(4)}_{23}=0.
\end{multline*} 
So (\ref{ybr}) is satisfied and that completes the proof.
\end{proof}
\begin{remark}
The Proposition above can be generalized. It is not necessary that the maps $R^{(i)},$ $i=2,3,4$ to be respectively the inverse and the companion maps of a quadrirational pentagon map $R^{(1)}$. Maps that satisfy    (\ref{pe_def}), (\ref{en1})-(\ref{en3}) and (\ref{yb1})-(\ref{yb4}), give rise to a Yang-Baxter map (\ref{PtoYB}).
\end{remark}

\begin{example}
Mapping $R^{(1)}:=\mathfrak{S}_I$ (that serves as a non-abelian form of mapping $S_I$ of Theorem \ref{theo1}) is
quadrirational where its inverse map $R^{(2)}:=(\mathfrak{S}_I)^{-1}$ and the companion maps
$R^{(3)}:=c\mathfrak{S}_I,$ $R^{(4)}:=(c\mathfrak{S}_I)^{-1}$ explicitly read
\begin{align*}
 R^{(2)}&:(x,y)\mapsto \left(xy,(1-x)y(1-xy)^{-1}\right),\\
 R^{(3)}&:(x,y)\mapsto \left((1-x+xy)^{-1} x y,x^{-1}(1-x+xy)^{-1}xy\right),\\
 R^{(4)}&:(x,y)\mapsto \left(xy^{-1},(y-x)(1-x)^{-1}\right).
\end{align*}
According to  Theorem \ref{theo3}, the mappings above satisfy the entwining relations (\ref{en1})-(\ref{en3}).
Furthermore, if we restrict for simplicity to the abelian setting, the Yang-Baxter map  (\ref{PtoYB}) associated with the mappings above explicitly reads $\mathfrak{R}_{\bf{12}}:(x,X;y,Y)\mapsto (u,U;v,V),$ where
\begin{align*} 
u=&xY\frac{1-X}{y(1-x)+Y(x-X)},& U=&\frac{XY}{y},& v=&\frac{y(1-x)+x-X}{1-X},&V=&Y\frac{y(1-x)+x-X}{y(1-x)+Y(x-X)}.
\end{align*}
%
This Yang-Baxter map is neither involutive (i.e. $\mathfrak{R}_{\bf{12}}^2\neq id$), nor reversible (i.e. $\tau_{\bf{12}}\circ\mathfrak{R}_{\bf{12}}\circ \tau_{\bf{12}}\circ\mathfrak{R}_{\bf{12}}\neq id$). However, it satisfies a {\em weak} reversibility property i.e. $\sigma_{\bf{12}}\circ\mathfrak{R}_{\bf{12}}\circ \sigma_{\bf{12}}\circ\mathfrak{R}_{\bf{12}}=id,$ where $\sigma_{\bf{12}}:(x,X;y,Y)\mapsto (X,x;Y,y).$
In addition, it is easy to verify that mapping $\mathfrak{R}_{\bf{12}}$ satisfies the following  invariant relations
\begin{align*}
\frac{x-X}{x(1-X)}=\frac{u-U}{u(1-U)}=&p, & \frac{y-Y}{y(1-Y)}=\frac{v-V}{v(1-V)}=&q,
\end{align*}
so in the new variables $(x,p;y,q)$ it becomes
\begin{align*}
\mathfrak{R}_{\bf{12}}:(x,p;y,q)\mapsto \left(\frac{(1+q)x}{1+qy+pqx(y-1)},p;y+px(y-1),q\right).
\end{align*}
 From mapping $\mathfrak{R}_{\bf{12}}$ under the conjugation of $\mathfrak{R}_{\bf{12}}$ with $\phi(p)\times\psi(q)$  where $\phi(p):x\mapsto \frac{x-1}{p}$ and $\psi(q):y\mapsto y+\frac{y-1}{q},$
we obtain $\widehat{\mathfrak{R}}_{\bf{12}}:=(\phi^{-1}(p)\times\psi(q)^{-1})\circ \mathfrak{R}_{\bf{12}}\circ (\phi(p)\times\psi(q))$ that reads
\begin{align*}
\widehat{\mathfrak{R}}_{\bf{12}}: (x,p;y,q)\mapsto \left(\frac{xy}{1-x+xy},p;1-x+xy,q\right),
\end{align*}
 that was obtained in \cite{ABS:YB} and served as a generic  Yang-Baxter  map of the so-called subclass $[1:1]$ of quadrirational Yang-Baxter maps. Note that two Yang-Baxter maps of the form $R,S:(x,p;y,q)\mapsto (f(x,y),p;g(x,y),q)$ are considered equivalent if there exists a bijection $\phi:\mathbb{X}\rightarrow \mathbb{X}$ such that $S=(\phi^{-1}(p)\times\phi(q)^{-1})\circ R\circ (\phi(p)\times\phi(q)),$ (see \cite{Papageorgiou:2010}). In that respect, mapping $\mathfrak{R}_{\bf{12}}$ is not equivalent to $\widehat{\mathfrak{R}}_{\bf{12}},$ 
but serves as an additional generic representative of an  equivalence class of maps inside the subclass $[1:1]$ of quadrirational Yang-Baxter maps.
\end{example}



\section*{Acknowledgements}
\parbox{.135\textwidth}{\begin{tikzpicture}[scale=.03]
\fill[fill={rgb,255:red,0;green,51;blue,153}] (-27,-18) rectangle (27,18);
\pgfmathsetmacro\inr{tan(36)/cos(18)}
\foreach \i in {0,1,...,11} {
\begin{scope}[shift={(30*\i:12)}]
\fill[fill={rgb,255:red,255;green,204;blue,0}] (90:2)
\foreach \x in {0,1,...,4} { -- (90+72*\x:2) -- (126+72*\x:\inr) };
\end{scope}}
\end{tikzpicture}} \parbox{.85\textwidth}
{This research is part of the project No. 2022/45/P/ST1/03998  co-funded by the National Science Centre and the European Union Framework Programme
 for Research and Innovation Horizon 2020 under the Marie Sklodowska-Curie grant agreement No. 945339. For the purpose of Open Access, the author has applied a CC-BY public copyright licence to any Author Accepted Manuscript (AAM) version arising from this submission.}


\end{document}